%% file: main.tex
\def\final{1}

\documentclass[11pt]{article}
	
\usepackage{amsthm, amsmath, amssymb}
\usepackage{mathtools}
\usepackage{bbm}
\usepackage{etex}
\usepackage{pgf}
\usepackage{enumerate}
\usepackage{enumitem}
\usepackage{framed}
\usepackage{verbatim}
\usepackage[left=1.0in, right=1.0in, top=1.0in, bottom=1.0in]{geometry}
\usepackage{microtype}
\usepackage[T1]{fontenc}
\usepackage{kpfonts}
	\DeclareMathAlphabet{\mathsf}{OT1}{cmss}{m}{n}
	\SetMathAlphabet{\mathsf}{bold}{OT1}{cmss}{bx}{n}
	\DeclareMathAlphabet{\mathtt}{OT1}{cmtt}{m}{n}
	\SetMathAlphabet{\mathtt}{bold}{OT1}{cmtt}{bx}{n}

\usepackage{multirow, tabularx}
\usepackage{algorithm, algorithmic}
\usepackage{color}
	\definecolor{DarkGreen}{rgb}{0.15,0.5,0.15}
	\definecolor{DarkRed}{rgb}{0.6,0.2,0.2}
	\definecolor{DarkBlue}{rgb}{0.15,0.15,0.55}
	\definecolor{DarkPurple}{rgb}{0.4,0.2,0.4}

\usepackage[pdftex]{hyperref}
\hypersetup{
	linktocpage=true,
	colorlinks=true,			
	linkcolor=DarkBlue,	
	citecolor=DarkBlue,	
	urlcolor=DarkBlue,		
	}
	
\newcolumntype{Y}{>{\centering\arraybackslash}X}

\ifnum\final=0
\newcommand{\mynote}[2]{{\color{#1} \marginpar{\tiny #2}}}
\newcommand{\mybignote}[2]{{\color{#1} $\langle \langle$ #2$\rangle \rangle$}}
\else
\newcommand{\mynote}[2]{}
\newcommand{\mybignote}[2]{}
\fi
\newcommand{\jon}[1]{\mynote{DarkRed}{Jon: {#1}}}

\newcommand{\mitali}[1]{\mynote{DarkBlue}{Mitali: {#1}}}

\newcommand{\p}[1]{\mathbb{P}\left[ #1 \right]}
\newcommand{\pr}[2]{\underset{#1}{\mathbb{P}}\left[ #2 \right]}
\newcommand{\e}[1]{\mathbb{E}\left[ #1 \right]}
\newcommand{\ex}[2]{\underset{#1}{\mathbb{E}}\left[ #2 \right]}

\newcommand{\zo}{\{0,1\}}

\newcommand{\pmo}{\{\pm1\}}

\newcommand{\set}[1]{\left\{#1\right\}}

\newcommand{\from}{:}
\newcommand{\tsum}{\textstyle \sum}

\renewcommand{\epsilon}{\varepsilon}
\newcommand{\eps}{\varepsilon}

\newcommand{\prob}{\operatorname{\mathbb{P}}}
\newcommand{\E}{\operatorname{\mathbb{E}}}
\newcommand{\brac}[1]{\left[ #1 \right]}
\newcommand{\paren}[1]{\left( #1 \right)}
\newcommand{\abs}[1]{\left\lvert #1 \right\rvert}

\newcommand{\iprod}[2]{\langle #1 , #2 \rangle}
\newcommand{\nn}{\mathbb{N}}

\newcommand{\ii}{\mathbb{I}}

\newcommand{\vecq}{\mathbf{q}}
\newcommand{\vect}{\mathbf{t}}

\newcommand{\vecx}{\mathbf{x}}
\newcommand{\vecy}{\mathbf{y}}

\newcommand{\N}{\mathbb{N}}

\newcommand{\cA}{\mathcal{A}}

\newcommand{\cM}{\mathcal{M}}

\newcommand{\cR}{\mathcal{R}}

\newcommand{\ds}{\{\pm{1}\}^{n \times d}}
\newcommand{\hatt}{\hat{\mathbf{t}}}
\newcommand{\IN}{\mathrm{IN}}
\newcommand{\OUT}{\mathrm{OUT}}
\newcommand{\attack}{\mathcal{A}}
\newcommand{\given}{\; \big | \; }

\newtheorem{theorem}{Theorem}[section]
\newtheorem{thm}[theorem]{Theorem}
\newtheorem{lemma}[theorem]{Lemma}
\newtheorem{lem}[theorem]{Lemma}

\newtheorem{claim}[theorem]{Claim}

\theoremstyle{definition}

\newtheorem{definition}[theorem]{Definition}

\title{The Price of Selection in Differential Privacy}
\author{Mitali Bafna\thanks{IIT Madras, Department of Computer Science and Engineering.  This research was performed while the author was a visiting scholar at Northeastern University.  \href{mailto:mitali.bafna@gmail.com}{mitali.bafna@gmail.com}} \and Jonathan Ullman\thanks{Northeastern University, College of Computer and Information Science.  \href{mailto:jullman@ccs.neu.edu}{jullman@ccs.neu.edu}}}

\begin{document}

\pagenumbering{gobble}
\maketitle
\begin{abstract}
In the differentially private top-$k$ selection problem, we are given a dataset $X \in \pmo^{n \times d}$, in which each row belongs to an individual and each column corresponds to some binary attribute, and our goal is to find a set of $k \ll d$ columns whose means are approximately as large as possible.  Differential privacy requires that our choice of these $k$ columns does not depend too much on any on individual's dataset.  This problem can be solved using the well known exponential mechanism and composition properties of differential privacy.  In the high-accuracy regime, where we require the error of the selection procedure to be to be smaller than the so-called sampling error $\alpha \approx \sqrt{\ln(d)/n}$, this procedure succeeds given a dataset of size $n \gtrsim k \ln(d)$.

We prove a matching lower bound, showing that a dataset of size $n \gtrsim k \ln(d)$ is necessary for private top-$k$ selection in this high-accuracy regime.  Our lower bound is the first to show that selecting the $k$ largest columns requires more data than simply estimating the value of those $k$ columns, which can be done using a dataset of size just $n \gtrsim k$.
\end{abstract}

\vfill
\newpage

\pagenumbering{arabic}

\input{intro.tex}

\input{top-k.tex}

\input{noisy-top-k}

\section*{Acknowledgements}
We are grateful to Adam Smith and Thomas Steinke for many helpful discussions about tracing attacks and private top-$k$ selection.

\bibliographystyle{alpha}
\bibliography{diffprivacy,refs}

\end{document}

%% file: intro.tex
\section{Introduction}
The goal of privacy-preserving data analysis is to enable rich statistical analysis of a sensitive dataset while protecting the privacy of the individuals who make up that dataset.  It is especially desirable to ensure \emph{differential privacy}~\cite{DworkMNS06}, which ensures that no individual's information has a significant influence on the information released about the dataset.  The central problem in differential privacy research is to determine precisely what statistics can be computed by differentially private algorithms and how accurately they can be computed.

The seminal work of Dinur and Nissim~\cite{DinurN03} established a ``price of privacy'':  If we release the answer to $\gtrsim n$ statistics on a dataset of $n$ individuals, and we do so with error that is asymptotically smaller than the \emph{sampling error} of $\approx 1/\sqrt{n}$, then an attacker can reconstruct nearly all of the sensitive information in the dataset, violating any reasonable notion of privacy.  For example, if we have a dataset $X = (x_1,\dots,x_n) \in \{\pm 1\}^{n \times d}$ and we want to privately approximate its \emph{marginal vector} $q = \frac{1}{n} \sum_{i=1}^{n} x_i$, then it is suffices to introduce error of magnitude $\Theta(\sqrt{d} / n)$ to each entry of $q$~\cite{DinurN03, DworkN04, BlumDMN05, DworkMNS06}, and this amount of error is also necessary~\cite{BunUV14, SteinkeU15b}.  Thus, when $d \gg n$, the error must be asymptotically larger the sampling error.

\paragraph{Top-k Selection.}
In many settings, we are releasing the marginals of the dataset in order to find a small set of ``interesting'' marginals, and we don't need the entire vector.  For example, we may be interested in finding only the attributes that are unusually frequent in the dataset.  Thus, an appealing approach to overcome the limitations on computing marginals is to find only the \emph{top-$k$} (approximately) largest coordinates of the marginal vector $q$, up to some error $\alpha$.\footnote{Here, the algorithm has error $\alpha$ if it returns a set $S \subseteq \set{1,\dots,d}$ consisting of of $k$ coordinates, and for each coordinate $j \in S$, $q_{j} \geq \tau - \alpha$, where $\tau$ is the $k$-th largest value among all the coordinates $\{q_1,\dots,q_{d}\}$.}  

Once we find these $k$ coordinates, we can approximate the corresponding marginals with additional error $O(\sqrt{k}/n)$.  But, how much error must we have in the top-$k$ selection itself?  The simplest way to solve this problem is to greedily find $k$ coordinates using the differentially private \emph{exponential mechanism}~\cite{McSherryT07}.  This approach finds the top-$k$ marginals up to error $\lesssim \sqrt{k} \log(d)/n$.  The \emph{sparse vector algorithm}~\cite{DworkNPR10, RothR10, HardtR10} would provide similar guarantees.

Thus, when $k \ll d$, we can find the top-$k$ marginals and approximate their values with much less error than approximating the entire vector of marginals.  However, the bottleneck in this approach is the $\sqrt{k} \log(d)/n$ error in the selection procedure, and this $\log(d)$ factor is significant in very high-dimensional datasets.  For comparison, the sampling error for top-$k$ selection is $\approx \sqrt{\log(d)/n}$ so the error introduced is asymptotically larger than the sampling error when $k \log(d) \gg n$.  However, the best known lower bound for top-$k$ selection follows by scaling down the lower bounds for releasing the entire marginal vector, and say that the error must be $\gtrsim \sqrt{k}/n$.

Top-$k$ selection is a special case of fundamental data analysis procedures like variable selection and sparse regression.  Moreover, private algorithms for selection problems underlie many powerful results in differential privacy: private control of false discovery rate~\cite{DworkS15}, algorithms for answering exponentially many queries~\cite{RothR10, HardtR10, GuptaRU12, JainT12, Ullman15}, approximation algorithms~\cite{GuptaLMRT10}, frequent itsemset mining~\cite{BhaskarLST10}, sparse regression~\cite{SmithT13}, and the optimal analysis of the generalization error of differentially private algorithms~\cite{BassilyNSSSU16}.  Therefore it is important to precisely understand optimal algorithms for differentially private top-$k$ selection.

Our main result says that existing differentially private algorithms for top-$k$ selection are essentially optimal in the high-accuracy regime where the error is required to be asymptotically smaller than the sampling error.

\begin{thm}[Sample Complexity Lower Bound for Approximate Top-$k$] \label{thm:main-lb}
There exist functions $n = \Omega(k \log(d))$ and $\alpha = \Omega(\sqrt{\log(d)/n})$ such that for every $d$ and every $k = d^{o(1)}$, there is no differentially private algorithm $M$ that takes an arbitrary dataset $X \in \pmo^{n \times d}$ and (with high probability) outputs an $\alpha$-accurate top-$k$ marginal vector for $X$.
\end{thm}

\paragraph{Tracing Attacks.}
Our lower bounds for differential privacy follow from a \emph{tracing attack}~\cite{H+08, SOJH09, BunUV14, SteinkeU15b, DworkSSUV15, DworkSSU17}.  In a tracing attack, the dataset $X$ consists of data for $n$ individuals drawn iid from some known distribution over $\pmo^{d}$.  The attacker is given data for a \emph{target individual} $y \in \pmo^{d}$ who is either one of the individuals in $X$ (``$\IN$''), or is an independent draw from the same distribution (``$\OUT$'').  The attacker is given some statistics about $X$ (e.g.~the top-$k$ statistics) and has to determine if the target $y$ is in or out of the dataset.  Tracing attacks are a significant privacy violation, as mere presence in the dataset can be sensitive information, for example if the dataset represents the case group in a medical study~\cite{H+08}.

Our results give a tracing attack for top-$k$ statistics in the case where the dataset is drawn uniformly at random.  For simplicity, we state the properties of our tracing attack for the case of the exact top-$k$ marginals.  We refer the reader to Section~\ref{sec:noisytopk} for a detailed statement in the case of approximate top-$k$ marginals, which is what we use to establish Theorem~\ref{thm:main-lb}.
\begin{thm}[Tracing Attack for Exact Top-$k$] \label{thm:main-attack-exact}
For every $\rho > 0$, every $n \in \nn$, and every $k \ll d \ll 2^n$ such that ${k \log(d/k) \geq O(n\log (1/\rho))}$, there exists an attacker $\mathcal{A}: \{-1,1\}^{d} \times \zo^{d} \to \set{\IN, \OUT}$ such that the following holds:  If we choose $X = (x_1,\dots,x_n) \in \pmo^{n \times d}$ uniformly at random, and $t(X)$ is the exact top-$k$ vector\footnote{Due to the presence of ties, there is typically not a unique top-$k$.  For technical reasons, and for simplicity, we let $t(X)$ denote the unique lexicographically first top-$k$ vector and refer to it as ``the'' top-$k$ vector.} of $X$, then  
\begin{enumerate}
\item If $y \in \pmo^{d}$ is uniformly random and independent of $X$, then
$\prob \brac{A(y,t(X)) = \OUT} \geq 1-\rho,$ and
\item  for every $i \in [n]$,
$\prob \brac{A(x_i,t(X)) = \IN } \geq 1 - \rho$.
\end{enumerate}
\end{thm}
While the assumption of uniformly random data is restrictive, it is still sufficient to provide a lower bound for differential privacy.  Tracing attacks against algorithms that release the entire marginal vector succeed under weaker assumptions---each column can have a different and essentially arbitrary bias as long as columns are independent.  However, for top-$k$ statistics, a stronger assumption on the column biases is necessary---if the column biases are such that $t(X)$ contains a specific set of columns with overwhelming probability, then $t(X)$ reveals essentially no information about $X$, so tracing will fail.  Under the weaker assumption that some unknown set of $k$ columns ``stand out'' by having significantly larger bias than other columns, we can use the propose-test-release framework~\cite{DworkL09} to find the exact top-$k$ vector when $n \gtrsim \log(d)$.  An interesting future direction is to characterize which distributional assumptions are sufficient to bypass our lower bound. 

We remark that, since our attack ``traces'' all rows of the dataset (i.e.~$\mathcal{A}(x_i, t(X)) = \IN$ for every $i \in [n]$), the attack bears some similarities to a \emph{reconstruction attack}~\cite{DinurN03, DworkMT07, DworkY08, KasiviswanathanRSU10, KasiviswanathanRS13, NikolovTZ13}.  However, the focus on high-dimensional data and the style of analysis is much closer to the literature on tracing attacks.

\subsection{Proof Overview}
Our results use a variant of the \emph{inner product attack} introduced in~\cite{DworkSSUV15} (and inspired by the work on \emph{fingerprinting codes}~\cite{BonehS95, Tardos03} and their connection to privacy~\cite{Ullman13, BunUV14, SteinkeU15b}).  Given a target individual $y \in \pmo^{d}$, and a top-$k$ vector $t \in \pmo^{d}$, the attack is
\begin{equation*}
\attack(y, t) =
\begin{cases}
\IN &\textrm{if $\langle y, t \rangle \geq \tau$} \\
\OUT &\textrm{otherwise}
\end{cases}
\end{equation*}
where $\tau = \Theta(\sqrt{k})$ is an appropriately chosen threshold.  The key to the analysis is to show that, when $X = (x_1,\dots,x_n) \in \pmo^{n \times d}$ and $y \in \pmo^{d}$ are chosen uniformly at random, $t(X)$ is an accurate top-$k$ vector of $X$, then
\begin{equation*}
\ex{}{\langle y, t(X) \rangle} = 0~~~~~\textrm{and}~~~~~\forall i \in [n]~~\ex{}{\langle x_i, t(X) \rangle} > 2\tau.
\end{equation*}
If we can establish these two facts then Theorem~\ref{thm:main-attack-exact} will follow from concentration inequalities for the two inner products.

Suppose $t(X)$ is the \emph{exact} top-$k$ vector.  Since each coordinate of $y$ is uniform in $\pmo$ and independent of $X$, we can write
$$
\e{\langle y, t(X) \rangle} = \textstyle\sum_{j} \e{y_j \cdot t(X)_j} = \textstyle\sum_{j} \e{y_j} \e{t(X)_j} = 0.
$$
Moreover, for every fixed vector $t \in \pmo^{d}$ with $k$ non-zero coordinates, $\langle y, t \rangle$ is a sum of $k$ independent, bounded random variables.  Therefore, by Hoeffding's inequality we have that $\langle y, t \rangle = O(\sqrt{k})$ with high probability.  Since $y, X$ are independent, this bound also holds with high probability when $X$ is chosen randomly and $t(X)$ is its top-$k$ vector.  Thus, for an appropriate $\tau = \Theta(\sqrt{k})$, $\attack(y, t(X)) = \OUT$ with high probability.

Now, consider the case where $y = x_i$ is a row of $X$, and we want to show that $\e{\langle x_i, t(X) \rangle}$ is sufficiently large.  Since $X$ is chosen uniformly at random.  One can show that, when $k \ll d \ll 2^n$, the top-$k$ largest marginals of $X$ are all at least $\gamma = \Omega(\sqrt{\log(d/k)/n})$.  Thus, on average, when $t(X)_j = 1$, we can think of $x_{i,j} \in \pmo$ as a random variable with expectation $\geq \gamma$.  Therefore,
$$
\e{\langle x_i, t(X) \rangle} = \e{\textstyle\sum_{j: t(X)_j = 1} x_{i,j} } \geq k \gamma = \Omega\left( k \sqrt{\log(d/k) / n} \right)
$$
Even though $x_i$ and $t(X)$ are not independent, and do not have independent entries, we show that with high probability over the choice of $X$, $\langle x_i, t(X) \rangle \geq k \sqrt{\log(d/k)/n} - O(\sqrt{k})$ with high probability.  Thus, if $k \log(d/k) \gtrsim n$, we have that $\attack(x_i,t(X)) = \IN$ with high probability.

\paragraph{Extension to Noisy Top-$k$.}
\jon{Need to change this slightly before we submit.}
The case of $\alpha$-approximate top-$k$ statistics does not change the analysis of $\langle y, t \rangle$ in that case that $y$ is independent of $x$, but does change the analysis of $\langle x_i, t \rangle$ when $x_i$ is a row of $X$.  It is not too difficult to show that for a \emph{random} row $x_i$, $\mathbb{E}[\langle x_i, \hat{t} \rangle] \gtrsim k(\gamma - \alpha)$, but it is not necessarily true that $\langle x_i, t \rangle$ is large for every row $x_i$.  The problem is that for relevant choices of $\alpha$, a random dataset has many more than $k$ marginals that are within $\alpha$ of being in the top-$k$, and the algorithm could choose a subset of $k$ of these to prevent a particular row $x_i$ from being traced.  For example, if there are $3k$ columns of $X$ that could be chosen in an $\alpha$-accurate top-$k$ vector, then with high probability, there exists a vector $t$ specifying $k$ of the columns on which $x_i = 0$, which ensures that $\langle x_i, t \rangle = 0$.

We can, however, show that $\langle x_i, \hat{t} \rangle > \tau$ for at least $(1-c)n$ rows of $X$ for an arbitrarily small constant $c > 0$.  This weaker tracing guarantee is still enough to rule out $(\eps, \delta)$-differential privacy for any reasonable setting of $\eps, \delta$ (Lemma~\ref{lem:tracingvsdp}), which gives us Theorem~\ref{thm:main-lb}.  The exact statement and parameters are slightly involved, so we refer the reader to Section~\ref{sec:noisytopk} for a precise statement and analysis of our tracing attack in the case of approximate top-$k$ statistics (Theorem~\ref{thm:noisytracing}).

%% file: top-k.tex
\section{Preliminaries}
\begin{definition}[Differential Privacy]
For $\epsilon \geq 0, \rho \in [0,1]$ we say that a randomized algorithm $\cM: \{\pm 1\}^{n \times d} \rightarrow \cR$ is $(\epsilon, \delta)$-differentially private if for every two datasets ${X, X' \in \{\pm 1\}^{n \times d}}$, such that $X, X'$ differ in at most one row, we have that, 
\begin{align*}
\forall S \subseteq \cR~~~~~\prob\brac{\cM(X) \in S} \leq e^\epsilon \cdot \prob\brac{\cM(X') \in S} + \delta. 
\end{align*}
\end{definition}

\begin{definition}[Marginals]
For a dataset $X \in \pmo^{n \times d}$, its \emph{marginal vector} $q(X) = (q_1(X), \ldots, q_d(X))$ is the average of the rows of $X$.  That is, $q_j(X) = \frac{1}{n} \sum_{i=1}^{n} X_{i,j}$. We use the notation 
\begin{align*}
q_{(1)}(X) \geq q_{(2)}(X) \geq \ldots \geq q_{(d)}(X)
\end{align*}
to refer to the \emph{sorted marginals}.  We will also define $\pi : [d] \to [d]$ to be the lexicographically first permutation that puts the marginals in sorted order.  That is, we define $\pi$ so that $q_{\pi(j)} = q_{(j)}$ and if $j < j'$ are such that $q_j = q_{j'}$, then $\pi(j) < \pi(j')$.
\end{definition}

\begin{definition}[Accurate Top-$k$ Vector] Given a dataset $X \in \ds$ and a parameter $\alpha \geq 0$, a vector $\hat{t} \in \{0,1\}^d$ is an \emph{$\alpha-$accurate top-$k$ vector of $X$} if,
	\begin{enumerate}
		\item 
		$\hat{t}$ has exactly $k$ non-zero coordinates, and
		\item
		$(\hat{t}_i = 1) \Rightarrow (q_i(X) \geq q_{(k)}(X) - \alpha).$
	\end{enumerate}
When $\alpha = 0$, we define the \emph{exact top-$k$ vector} of $X$ as $t(X) \in \{0,1\}^d$ to be the lexicographically first $0$-accurate top-$k$ vector.\footnote{Due to ties, there may not be a unique $0$-accurate top-$k$ vector of $X$.  For technical reasons we let $t(X)$ be the unique lexicographically first $0$-accurate top-$k$ vector, so we are justified in treating $t(X)$ as a function of $X$.}  Specifically, we define $t(X)$ so that
 $$(t(X)_j = 1) \Leftrightarrow j \in \{\pi(1),\ldots, \pi(k)\}.$$ We refer to these set of columns as the top-k columns of $X$. 
\end{definition}

For comparison with our results, we state a positive result for privately releasing an $\alpha$-approximate top-$k$ vector, which is an easy consequence of the exponential mechanism and composition theorems for differential privacy
\begin{theorem}
For every $n,d,k \in \N$, and $\eps, \delta, \beta \in (0,1)$, there is an $(\eps, \delta)$-differentially private algorithm that takes as input a dataset $X \in \pmo^{n \times d}$, and with probability at least $1-\beta$, outputs an $\alpha$-accurate top-$k$ vector of $X$, for
$$
\alpha = O\left(\frac{\sqrt{k \cdot \ln(1/\delta)} \cdot \ln(kd/\beta)}{\eps n}\right)
$$
\end{theorem}


\subsection{Tracing Attacks}

Intuitively, tracing attacks violate differential privacy because if the target individual $y$ is outside the dataset, then $\cA(y, M(X))$ reports $\OUT$ with high probability, whereas if $y$ were added to the dataset to obtain $X'$, $\cA(y, M(X'))$ reports $\IN$ with high probability.  Therefore $M(X), M(X')$ must have very different distributions, which implies that $M$ is not differentially private.  The next lemma formalizes and quantifies this property.
\begin{lem}[Tracing Violates DP] \label{lem:tracingvsdp}
Let $M \from \pmo^{n \times d} \to \cR$ be a (possibly randomized) algorithm.  Suppose there exists an algorithm $\cA \from \{\pm 1\}^d \times \cR \rightarrow \set{\IN, \OUT}$ such that when $X \sim \pmo^{n \times d}$ and $y \sim \pmo^{d}$ are independent and uniformly random,
\begin{enumerate}
\item (Soundness) $\p{\cA(y, M(X)) = \IN} \leq \rho$
\item (Completeness) $\p{\#\set{ i \mid \cA(x_i, M(X)) = \IN} \geq n-m} \geq 1-\rho$.
\end{enumerate}
Then $M$ is not $(\eps, \delta)$-differentially private for any $\eps, \delta$ such that $e^{\eps}\rho + \delta < 1 - \rho - \frac{m}{n}$.  If $\rho < 1/4$ and $m > 3n/4$, then there are absolute constants $\eps_0, \delta_0 > 0$ such that $M$ is not $(\eps_0, \delta_0)$-differentially private.
\end{lem}
The constants $\eps_0, \delta_0$ can be made arbitrarily close to $1$ by setting $\rho$ and $m/n$ to be appropriately small constants.  Typically differentially private algorithms typically satisfy $(\eps, \delta)$-differential privacy where $\eps = o(1), \delta = o(1/n)$, so ruling out differential privacy with constant $(\eps, \delta)$ is a strong lower bound.

\subsection{Probabilistic Inequalities}
We will make frequent use of the following concentration and anticoncentration results for sums of independent random variables.
\begin{lemma}[Hoeffding Bound]\label{Hoeffding}
Let $Z_1, \ldots , Z_n$ be independent random variables supported on $\pmo$, and let $Z = \frac{1}{n} \sum_{i=1}^{n} Z_i$.  Then
\begin{align*}
\forall \nu > 0~~~~~\prob\brac{Z - \E[Z] \geq \nu} \leq e^{-\frac{1}{2}\nu^2 n}.
\end{align*}
\end{lemma}

Hoeffding's bound on the upper tail also applies to random variables that are \emph{negative-dependent}, which in this case means that setting any set of the variables $B$ to $+1$ only makes the variables in $[n] \setminus B$ more likely to be $-1$~\cite{PS97}.  Similarly, if the random variables are \emph{positive-dependent} (their negations are negative-dependent), then Hoeffding's bound applies to the lower tail.

\begin{theorem}[Chernoff Bound]\label{Chernoff}
Let $Z_1, \ldots, Z_n$ be a sequence of independent $\zo$-valued random variables, let $Z = \sum_{i=1}^{n} Z_i$, and let $\mu = \E[Z]$.  Then
\begin{enumerate}
\item
(Upper Tail)~~~ $\forall \nu > 0~~~\prob\brac{Z \geq (1+\nu)\mu} \leq e^{-\frac{\nu^2}{2+\nu}\mu}$, and
\item 
(Lower Tail)~~~ $\forall \nu \in (0,1)~~~\prob\brac{Z \leq (1-\nu)\mu} \leq e^{- \frac{1}{2} \nu^2 \mu}.$
\end{enumerate}
\end{theorem}

\begin{theorem}[Anticoncentration~\cite{LedouxT13}]\label{anti-conc} 
Let $Z_1,\dots,Z_n$ be independent and uniform in $\pmo$, and let $Z = \frac{1}{n} \sum_{i=1}^{n} Z_i$.  Then for every $\beta > 0$, there exists $K_{\beta} > 1$ such that for every $n \in \N$,
$$
\forall v \in \left[\frac{K_{\beta}}{\sqrt{n}} , \frac{1}{K_{\beta}}\right]~~~~~\mathbb{P}\left[Z \geq \nu\right] \geq e^{-\frac{1+\beta}{2} \nu^2 n}.
$$
\end{theorem}

\section{Tracing Using the Top-k Vector} \label{sec:exacttracing}

Given a (possibly approximate) top-$k$ vector $t$ of a dataset $X$, and a target individual $y$, we define the following \emph{inner product attack}.

\begin{center}
\fbox{ \label{alg:attack}
	\parbox{11 cm}{
		\begin{align*}
		\mathcal{A}_{\rho , d, k}(y, t):
		\end{align*}				
		
		\begin{enumerate}
		\item[] 		
		Input: $y \in \{\pm 1\}^d$ and $t \in \zo^d$.  Let $\tau = \sqrt{2k \ln(1/\rho)}.$
		\item[] 
		If $\iprod{y}{t} > \tau$, output IN; else output OUT.
		\end{enumerate}
	}
}
\end{center}
\bigskip

In this section we will analyze this attack when $X \in \pmo^{n \times d}$ is a uniformly random matrix, and $t = t(X)$ is the exact top-$k$ vector of $X$. 
In this case, we have the following theorem.
\begin{theorem} \label{thm:tracingexact}
There is a universal constant $C \in (0,1)$ such that if $\rho > 0$ is any parameter and $n,d,k \in \N$ satisfy $d \leq 2^{C n}$, $k \leq C d$ and $k\ln(d/2k) \geq 8n\ln(1/\rho)$, then $\mathcal{A}_{\rho,d,k}$ has the following properties:
If $X \sim \pmo^{n \times d}, y \sim \pmo^{d}$ are independent and uniform, and $t(X)$ is the exact top-k vector of $X$, then
\begin{enumerate}
\item (Soundness)~~
$\prob \brac{\mathcal{A}_{\rho,d,k}(y,t(X)) = \mathrm{IN}} \leq \rho $, and 
\item (Completeness)~~
for every $i \in [n]$,~$\prob \brac{ \mathcal{A}_{\rho,d,k}(x_i,t(X)) = \mathrm{OUT} } < \rho + e^{-k/4}$.
\end{enumerate}
\end{theorem}

We will prove the soundness and completeness properties separately in Lemmas~\ref{sound_top} and~\ref{comp_top}, respectively.  The proof of soundness is straightforward.
\begin{lemma}[Soundness]\label{sound_top} 
	For every $\rho > 0$, $n \in \N$, and $k \leq d \in \N$, if $X \sim \pmo^{n \times d}, y \sim \pmo^{d}$ are independent and uniformly random, and $\vect(X)$ is the exact top-$k$ vector, then
	\begin{align*}
		\mathbb{P}\brac{\iprod{y}{t(X)}  \geq \sqrt{ 2k\ln (1/\rho)}} \leq \rho. 
	\end{align*}
\end{lemma}

\begin{proof}
	%
	
	Recall that $\tau := \sqrt{2k\ln(1/\rho)}$.  Since $X,y$ are independent, we have
	
	
	\begin{align*}
		&\underset{X,y }{\prob}\brac{\iprod{y}{t(X)} \geq \tau} \\
		= \sum_{T \subseteq [d] : |T|=k}~&\underset{X, y}{\prob}\brac{\iprod{y}{t(X)} \geq \tau \; \big | \; t(X) = \ii_T} \cdot \pr{X}{t(X) = \ii_T}\\
		=  \sum_{T \subseteq [d] : |T|=k}~&\underset{y}{\prob}\left[\sum_{j \in T} y_j \geq \tau \right] \cdot \pr{X}{t(X) = \ii_T} \tag{$X, y$ are independent}\\
	\leq \max_{T \subseteq [d] : |T|=k}~&\underset{y}{\prob}\left[\sum_{j \in T} y_j \geq \tau \right]\\
	\end{align*}

	\noindent For every fixed $T$, the random variables $\{y_j\}_{j \in T}$ are independent and uniform on $\pmo$, so by Hoeffding's inequality,
	$${\p{ \sum_{j \in T} y_j \geq \sqrt{ 2k\ln (1/\rho)}} \leq \rho}.$$ This completes the proof of the lemma.
\end{proof}

\jon{Changed the constant from $\kappa$ to $C$ to be consistent with the noisy section.}
We now turn to proving the completeness property, which will following immediately from the following lemma.
\begin{lemma}[Completeness]\label{comp_top}
There is a universal constant $C \in (0,1)$ such that for every $\rho > 0$, $n \in \N$, $d \leq 2^{C n}$, and $k \leq C d$, if $X \sim \pmo^{n \times d}$ is chosen uniformly at random, $t(X)$ is the exact top-$k$ vector, and $x_i$ is any row of $X$,
\begin{align*}
\prob\brac{\iprod{x_i}{t(X)}  \leq k\sqrt{\frac{\ln(d/2k)}{n}} - \sqrt{ 2k\ln (1/\rho)} } \leq \rho + e^{-k/4}.
\end{align*}
\end{lemma}

To see how the completeness property of Theorem~\ref{thm:tracingexact} follows from the lemma, observe if $k \ln(d/2k) \geq 8n\ln(1/\rho)$, then $k\sqrt{\ln(d/2k)/n} - \sqrt{2k\ln(1/\rho)} \geq \tau$.  Therefore Lemma~\ref{comp_top} implies that $\p{\langle x_i, t(X) \rangle < \tau} \leq \rho + e^{-k/4}$, so $\p{\mathcal{A}_{\rho, d, k}(x_i, t(X)) = \IN} \geq 1-\rho-e^{-k/4}$.

Before proving the lemma, we will need a few claims about the distribution of $\langle x_i, t(X) \rangle$.  The first claim asserts that, although $X \in \pmo^{n \times d}$ is uniform, the $k$ columns of $X$ with the largest marginals are significantly biased.
\begin{claim} \label{clm:topkbias}
There is a universal constant $C \in (0,1)$, such that for every $n \in \N$, $d \leq 2^{C n}$ and $k \leq C d$, if $X \in \pmo^{n \times d}$ is drawn uniformly at random, then
$$
\p{q_{(k)}(X) < \sqrt{\frac{\ln(d/2k)}{n}} }  \leq e^{-k/4}.
$$
\end{claim}

\begin{proof}[Proof of Claim~\ref{clm:topkbias}]
For every $j \in [d]$, define $E_j$ to be the event that $$q_j = \frac{1}{n}\underset{i \in [n]}{\sum} x_{ij} > \sqrt{\frac{\ln(d/2k)}{n}}.$$

We would like to apply Theorem~\ref{anti-conc} to the random variable $\frac{1}{n}\sum_{i} x_{ij}$.  To do so, we need $$\sqrt{\ln\paren{d/2k}/n} \in \left[ \frac{K_1}{\sqrt{n}}, \frac{1}{K_{1}}\right]$$ where $K_{1}$ is the universal constant from that theorem (applied with $\beta = 1$).  These inequalities will be satisfied as long as $d \leq 2^{C n}$, and $k \leq C d$ for a suitable universal constant $C \in (0,1)$.  Applying Theorem~\ref{anti-conc} gives
$$\forall j \in [d]~~~~\p{E_j} =  \p{q_j  > \sqrt{\frac{\ln(d/2k)}{n}}} \geq \frac{2k}{d}.$$

By linearity of expectation, we have that $\E[\sum_j E_j] \geq 2k$. Since the columns of $X$ are independent, and the events $E_j$ only depend on a single column of $X$, the events $E_j$ are also independent.  Therefore, we can apply a Chernoff bound (Theorem~\ref{Chernoff}) to  $\sum_j E_j$ to get
$$\p{\sum_{j=1}^{d} E_j < k} \leq e^{-k/4}.$$

If $\sum_j E_j \geq k$, then there exist $k$ values $q_j$ that are larger than $\sqrt{\frac{\ln(d/2k)}{n}}$, so $q_{(k)}$ is also at least this value.  This completes the proof of the claim.
\end{proof}

The previous claim establishes that if we restrict $X$ to its top-$k$ columns, the resulting matrix $X_{t} \in \pmo^{n \times k}$ is a random matrix whose mean entry is significantly larger than $0$.  This claim is enough to establish that the inner product $\langle x_i, t(X) \rangle$ is large in expectation over $X$.  However, since $X_{T}$ its columns are not necessarily independent, which prevents us from applying concentration to get the high probability statement we need.  However, the columns of $X_{t}$ are independent if we condition on the value and location of the $(k+1)$-st marginal.

\begin{claim}\label{clm:independence}
	
	
	Let $X \in \pmo^{n \times d}$ be a random matrix from a distribution with independent columns, and let $t(X)$ be its marginals.  For every $q \in [-1,1], k, j \in [d], T \in \binom{[d]}{k}$, the conditional distribution $$X \mid (q_{(k+1)} = q) \land (\pi(k+1) = j) \land (t(X) = \mathbb{I}_T)$$ also has independent columns. 
\end{claim}
\mitali{Do we need to write a more formal proof than this?}
\jon{No I think it's a great level of detail!  I just made some small changes to the wording.  For example I changed $i$ to $\ell$ since we generally use $i$ for rows not columns.}
\begin{proof}[Proof of Claim~\ref{clm:independence}]
Suppose we condition on the value of the $(k+1)$-st marginal, $q_{(k+1)} = q$, its location, $\pi(k+1) = j$, and the set of top-$k$ marginals $t = \mathbb{I}_{T}$.  By definition of the (exact, lexicographically first) top-$k$ vector, we have that if $\ell < j$, then $\ell \in T$ if and only if $q_{\ell} \geq q$.  Similarly, if $\ell > j$, then $\ell \in T$ if and only if $q_\ell > q$.  Since we have conditioned on a fixed tuple $(q,j,T)$, the statements $q_\ell > q$ and $q_\ell \geq q$ now depend only on the $\ell$-th column.  Thus, since the columns of $X$ are independent, they remain independent even when condition on any tuple $(q,j,T)$.  Specifically, if $\ell < j$ and $\ell \in T$, then column $\ell$ is drawn independently from the conditional distribution $((u_1,\dots,u_n) \mid \frac{1}{n} \sum_i u_i \geq q)$, where $(u_1,\dots,u_n) \in \pmo^{n}$ are chosen independently and uniformly at random.  Similarly, if $\ell > j$ and $\ell \in T$, then column $\ell$ is drawn independently from $((u_1,\dots,u_n) \mid \frac{1}{n} \sum_i u_i > q)$.
\end{proof}

Now we are ready to prove Lemma~\ref{comp_top}.
\begin{proof}[Proof of Lemma~\ref{comp_top}]
For convenience, define $\gamma = \sqrt{\frac{\ln(d/2k)}{n}}$ and $\tau_c = k\gamma - \sqrt{ 2k\ln (1/\rho)}.$  Fix any row $x_i$ of $X$.  We can write
\begin{align}
&\prob\brac{{\iprod{x_i}{t(X)} < \tau_c}} \nonumber \\
\leq{} &\prob\left[ {\iprod{x_i}{t(X)} < \tau_c} \; \big | \; q_{(k+1)} \geq \gamma \right] + \prob\brac{q_{(k+1)} < \gamma} \nonumber \\
\leq{} &\prob\left[ {\iprod{x_i}{t(X)} < \tau_c} \; \big | \; q_{(k+1)} \geq \gamma \right] + e^{-k/4}\tag{Claim~\ref{clm:topkbias}} \nonumber \\
\leq{} &\max_{q \geq \gamma, j \in [d], T \in \binom{[d]}{k}} \p{\langle x_i, t(X) \rangle < \tau_c \big | \; (q_{(k+1)} = q) \land (\pi(k+1) = j) \land (t(X) = \mathbb{I}_{T})} + e^{-k/4} \label{eq:comp_main}
\end{align}
Let $G_{q, j, T}$ be the event $(q_{(k+1)} = q) \land (\pi(k+1) = j) \land (t(X) = \mathbb{I}_{T})$.
By linearity of expectation, we can write
$$
\e{\langle x_i, t(X) \rangle < \tau_c \big | \; G_{q, j, T}} \geq kq \geq k\gamma.
$$
Using Claim~\ref{clm:independence}, we have that $\sum_{\ell: t(X)_\ell = 1} x_{i\ell}$ conditioned on $G_{\delta, j, T}$ is a sum of independent $\pmo$-valued random variables.  Thus,
\begin{align*}
\p{\langle x_i, t(X) \rangle < \tau_c \mid G_{q, j, T}} 
={} &\p{\langle x_i, t(X) \rangle < k\gamma - \sqrt{2 k \ln(1/\rho)} \mid G_{q, j, T}} \\
\leq{} &\p{\langle x_i, t(X) \rangle < kq - \sqrt{2 k \ln(1/\rho)} \mid G_{q, j, T}} \tag{$q \geq \gamma$} \\
\leq{} &\rho \tag{Hoeffding}
\end{align*}

Combining with~\eqref{eq:comp_main} completes the proof.
\end{proof}

%% file: noisy-top-k.tex
\section{Tracing Using an Approximate Top-k Vector} \label{sec:noisytopk}
In this section we analyze the inner product attack when it is given an arbitrary approximate top-$k$ vector.

\begin{theorem} \label{thm:noisytracing}
For every $\rho > 0$, there exist  universal constants $C,C' \in (0,1)$ (depending only on $\rho$) such that if $d \in \N$ is sufficiently large and $n,d,k \in \N$ and $\alpha \in (0,1)$ satisfy 
$$
d \leq 2^{C n},~~~~
k\leq d^C,~~~~
n = C' k \ln (d/2k),~~~~
\textrm{and}~~~~\alpha \leq C \sqrt{\frac{\ln(d/2k)}{n}},
$$ 
and $\hat{t} \from \pmo^{n \times d} \to \zo^{d}$ is any randomized algorithm such that 
$$
\forall X \in \pmo^{n \times d}~~~~\p{\hat{t}(X) \textrm{ is an $\alpha$-approximate top-$k$ vector for $X$}} \geq 1-\rho,
$$
then $\mathcal{A}_{\rho, d, k}$ (Section~\ref{sec:exacttracing}) has the following properties: 
If $X \sim \pmo^{n \times d}, y \in \pmo^{d}$ are independent and uniform, then
\begin{enumerate}
\item (Soundness)~~
$\prob \brac{\mathcal{A}_{\rho,d,k}(y,t(X)) = \mathrm{IN}} \leq \rho$, and 
\item (Completeness)~~
$\prob \brac{\#\set{ i \in [n]~\left|~\mathcal{A}_{\rho, d, k}(x_i, \hat{t}(X)) = \mathrm{IN} \right.} < (1-e^2 \rho)n } < 2\rho + 2e^{-k/6}$.
\end{enumerate}
\end{theorem}

The proof of the soundness property is nearly identical to the case of exact statistics so we will focus only on proving the completeness property.   

Intuitively, the proof of completeness takes the same general form as it did for the case of exact top-$k$ statistics.  First, for some parameter $\gamma > 0$, with high probability $X$ has at least $k$ marginals that are at least $\gamma$.  Therefore, any marginal $j$ contained in \emph{any} $\alpha$-accurate top-$k$ vector has value $q_j \geq \lambda := \gamma - \alpha$.  From this, we can conclude that the expectation of $\langle x_i, \hat{t}(X) \rangle \geq k(\gamma- \alpha)$ where the expectation is taken over the choices of $X, \hat{t}(X),$ and $i$.  However, unlike the case of the exact marginals, the $k$ columns selected by $\hat{t}$ may be significantly correlated so that for some choices of $i$, $\langle x_i, \hat{t}(X) \rangle$ is small with high probability.  At a high level we solve this problem as follows: first, we restrict to the set of $d_{\lambda}$ columns $j$ such that $q_j \geq \gamma - \alpha$, which remain mutually independent.  Then we argue that for every fixed $\alpha$-accurate top-$k$ vector specifying a subset of $k$ of these columns, with overwhelming probability the inner product is large for most choices of $i \in [n]$.  Finally, we take a union bound over all $\binom{d_{\lambda}}{k}$ possible choices of $\alpha$-accurate top-$k$ vector.  To make the union bound tolerable, we need that with high probability $d_{\lambda}$ is not too big.  Our choice of $\gamma$ was such that only about $k$ columns are above $\gamma$, therefore if we take $\lambda$ very close to $\gamma$, we will also be able to say that $d_{\lambda}$ is not too much bigger than $k$.  By assuming that the top-$k$ vector is $\alpha$-accurate for $\alpha \ll \gamma$, we get that $\lambda = \gamma - \alpha$ is very close to $\gamma$.

Before stating the exact parameters and conditions in Lemma~\ref{comp_noisy}, we will need to state and prove a few claims about random matrices.

\begin{claim} \label{lem:noisytopkbias}
For every $\beta > 0$, there is a universal constant $C \in (0,1)$ (depending only on $\beta$), such that for every $n \in \N$, $d \leq 2^{C n}$ and $k \leq C d$, if $X \in \pmo^{n \times d}$ is drawn uniformly at random, then for $\gamma := \sqrt{\frac{2}{1+\beta} \cdot \frac{\ln(d/2k)}{n}}$, we have
$$
\p{q_{(k)}(X) < \gamma }  \leq e^{-k/4}.
$$
\end{claim}

The above claim is just a slightly more general version of Claim~\ref{clm:topkbias} (in which we have fixed $\beta = 1$), so we omit its proof.

\begin{claim}\label{clm:1}
	For every $n,d \in \N$ and every $\lambda \in (0,1)$, if $X \in \ds$ is drawn uniformly at random, then for $d_{\lambda} := 2d \exp(-\frac{1 }{ 2} \lambda^2 n)$
	\begin{align*}
		\p{\# \set{ j~\left|~q_j > \lambda \right.} > d_{\lambda}} \leq e^{-d_{\lambda}/6}.
	\end{align*} 
\end{claim}

\begin{proof}[Proof of Claim~\ref{clm:1}]
	For every $j \in [d]$, define $E_j$ to be the event that $q_j = \frac{1}{n}\sum_{i} x_{ij} > \lambda.$	
 	Since the $x_{ij}$'s are independent, applying Hoeffding's bound to $\sum_i x_{ij}$ gives,
	$$\forall j \in [d]~~~~\p{E_j} =  \p{q_j  > \lambda} \leq e^{-\lambda^2 n /2}.$$
	
	By linearity of expectation, we have that $\E[\sum_j E_j] \leq de^{-\lambda^2n / 2} = \frac{1}{2}d_{\lambda}$. Since the columns of $X$ are independent, we can apply a Chernoff bound (Theorem~\ref{Chernoff}) to $\sum_j E_j$, which gives
	$$\p{\textstyle\sum_{j=1}^{d} E_j > d_{\lambda}} \leq e^{-d_{\lambda}/6}.$$
	
	This completes the proof of the claim.
\end{proof}
 

\jon{It's ever-so-slightly weird to have constants $C_2,\dots,C_5$ instead of $C_1,\dots,C_4$ but it doesn't seem worth the trouble to fix now.}
\mitali{It was a Constant dilemna I had too! But $C_1$ has to be the first constant we define in the proof whereas it doesn't appear in the lemma at all.}
\jon{Pun intended?}

Now we are ready to state our exact claim about the completeness of the attack when given an $\alpha$-accurate top-$k$ vector.



\begin{lemma}[Completeness]\label{comp_noisy} 
	For every $\rho > 0$, there exist universal constants $C_2,C_3,C_4,C_5 \in (0,1)$ (depending only on $\rho$) such that if $n,d,k \in \N$ and $\alpha \in (0,1)$ satisfy, 
	$${4k \leq \min\{(2d)^{C_2}, 4C_4 d\}}, ~~~ 
	{8n \ln(1/\rho) = C_3^2 k \ln(2d)}, ~~~
	d \leq 2^{C_4 n} ~~~
	\alpha \leq C_5\sqrt{\frac{\ln(2d)}{n}},~~~
	$$
	 and $\hat{t}$ is an algorithm that, for every $X \in \pmo^{n \times d}$, outputs an $\alpha$-accurate top-$k$ vector with probability at least $1-\rho$, then for a uniformly random $X \in \pmo^{n \times d}$, we have
	$$\prob\brac{\# \left\{i \in [n] \mid \iprod{x_i}{\hat{t}(X)} \geq \tau_c \right\} < (1-e^2 \rho)n }  < 2\rho + e^{-k/4} + e^{-k/6},$$
	where $\tau_c := C_3k\sqrt{\frac{\ln(2d)}{n}} - \sqrt{2k\ln(1/\rho)}$. 
\end{lemma}
To see how the completeness property of Theorem~\ref{thm:noisytracing} follows from the lemma, observe if ${8n \ln(1/\rho) = C_3^2 k \ln(2d)}$, then $$\tau_c = C_3k\sqrt{\frac{\ln(2d)}{n}} - \sqrt{2k\ln(1/\rho)} = \sqrt{2k\ln(1/\rho)} = \tau$$ where $\tau$ is the threshold in $\mathcal{A}_{\rho, d,k}$.  Therefore Lemma~\ref{comp_noisy} implies that $$\prob \brac{\#\set{ i \in [n]~\left|~\mathcal{A}_{\rho, d, k}(x_i, \hat{t}(X)) = \mathrm{IN} \right.} < (1-e^2 \rho)n } < \rho + e^{-k/4} + e^{-k/6}.$$ 

The universal constants $C,C'$ will be $C = \min\{ C_2,C_4,C_5 \} - \delta$ for an arbitrarily small $\delta > 0$, and $C' = C_3^2$.  As long as $d$ is sufficiently large the conditions $k \leq d^{C}$ in Theorem~\ref{thm:noisytracing} will imply the corresponding condition in the above lemma.

\begin{proof}[Proof of Lemma~\ref{comp_noisy}]

First, we will condition everything on the event $$G_{\alpha} := \{\textrm{$\hat{t} = \hat{t}(X)$ is an $\alpha$-accurate top-$k$ vector of $X$}\}.$$  By assumption, for every $X \in \pmo^{n \times d}$, $\p{G_\alpha} \geq 1-\rho$.

\jon{This is low-priority, but it might make sense to get rid of $c$ completely and just write $e^2 \rho$.  Since the constants are messy and hard-to-understand anyway, I'm not sure it's worth it to make the reader translate between $c$ and $\rho$.}

\medskip

For convenience define the constant $c := e^2\rho$, so that the lemma asserts that, with high probability, $\mathcal{A}(x_i, \hat{t}(X)) = \IN$ for at least $(1-c)n$ rows $x_i$.  
As in the proof of completeness for the case of exact top-$k$, we will first condition on the event that at least $k$ marginals are above the threshold $\gamma$.  Now, by Claim~\ref{lem:noisytopkbias}, with an appropriate choice of 
$$
\beta := \frac{c}{16\ln(1/\rho)}~~~~~
\gamma :=  \sqrt{\frac{2}{1+ \frac{c}{16\ln(1/\rho)}}} \cdot \sqrt{\frac{\ln(d/2k)}{n}},
$$ 
and the assumptions that $k \leq C_4 d$ and $d \leq 2^{C_4 n}$ for some universal constant $C_4$ depending only on $\beta$, the event
$$
G_{\gamma} := \left\{ q_{(k)}(X) \geq \gamma = C_1 \sqrt{\frac{\ln(2d)}{n}} \right\},
$$
will hold with probability $1-e^{-k/4}$.  Here we define the universal constants
$$
C_1 := \sqrt{\frac{2}{1 + \frac{c}{8\ln(1/\rho)}}}~~~~C_2 := \frac{c}{2c + 4 \ln(1/\rho)}.
$$
depending only on $\rho$.  These constants were chosen so that provided $4k \leq (2d)^{C_2}$, the inequality in the definition of $G_{\gamma}$ will be satisfied.

In light of the above analysis, we condition the rest of the analysis on the event $G_{\gamma}$, which is satisfies $\p{G_{\gamma}}\geq 1 - e^{-k/4}$.

\medskip

If we condition on $G_{\alpha}$ and $G_{\gamma}$, then for any marginal $j$ chosen by $\hat{t}$ (i.e.~$\hat{t}_j = 1$), then we can say that $q_{j} \geq \lambda$ for any $\lambda \leq \gamma - \alpha$.  Now, we define the constants 
$$
C_3 := \sqrt{\frac{2}{1 + \frac{c}{4\ln(1/\rho)}}}~~~~C_5 := C_1 - C_3 > 0,
$$
where one can verify that the inequality $C_1 - C_3 > 0$ holds for all choices of $c$.  Now by our assumption that $\alpha < C_5 \sqrt{\frac{\ln(2d)}{n}}$, we can define
$
\lambda := C_3 \sqrt{\frac{\ln(2d)}{n}}.
$

For any matrix $X \in \pmo^{n \times d}$, we can define $S_\lambda = S_{\lambda}(X) \subseteq \set{1,\dots,d}$ to be the set of columns of $X$ whose marginals are greater than $\lambda$. The analysis above says that, conditioned on$G_{\gamma}$ and $G_{\alpha}$, if $\hat{t}_j = 1$, then $j \in S_{\lambda}$.  Note that, if $X$ is chosen uniformly at random, and we define $X_{\geq \lambda} \in \pmo^{n \times |S_{\lambda}|}$ to be the restriction of $X$ to the columns contained in $S_{\lambda}$, then the columns of $X_{\geq \lambda}$ remain independent.

The size of $S_{\lambda}$ is a random variable supported on $\{0,1,\dots,d\}$.  In our analysis we will need to condition on the even that $|S_{\lambda}| \ll d$.  Using Claim~\ref{clm:1} we have that if
$$
d_{\lambda} := 2d e^{-\frac{\lambda^2 n}{2}}
$$
then the event
$$
G_{S} := \left\{ |S_{\lambda}(X)| \leq d_{\lambda} \right\}
$$
satisfies $\p{G_{S}} \geq 1 - e^{-d_{\lambda}/6} \geq 1 - e^{-k/6}$ where we have used the fact that $d_{\lambda} \geq k$.  This fact is not difficult to verify form our choice of parameters.  Intuitively, since $\lambda \leq \gamma$, and there are at least $k$ marginals larger than $\gamma$, there must also typically be at least $k$ marginals larger than $\lambda.$  We condition the remainder of the analysis on the event $G_{S}$.

Later in the proof we require that the size of $S_\lambda$ is small with high probability. Using Claim~\ref{clm:1} we can say that the size of $S_\lambda$ is at most $d_\lambda = 2de^{-\frac{\lambda^2n}{2}}$ with probability at least, $1 - e^{-d_\lambda/6}$. When $\vecq_{(k)} \geq \gamma$,  the number of marginals greater than $\lambda$ would be at least $k$. So $d_\lambda > k$ and the error probability $e^{-d_\lambda/6}$ is at most $e^{-k/6}$. We will henceforth condition on the event that $|S_\lambda(X)| \leq d_\lambda$.

We will say that the attack $\attack$ \emph{fails on $\hat{t}$} when we fail to trace more than $cn$ rows, i.e. $\attack$ fails when ${\abs{\{i : \iprod{x_i}{\hat{t}} < k\lambda - \sqrt{2k\ln(1/\rho)} \} } > cn = e^2 \rho n}$. Formally we have that,


\begin{align}\label{eqn1}
	\p{\attack \textrm{ fails on $\hat{t}$}} 
	\leq{} &\p{\attack \textrm{ fails on $\hat{t}$} \land G_{\alpha} \land G_{\gamma} \land G_{S}} + \p{\neg G_\alpha \lor \neg G_{\gamma} \lor \neg G_{S}}  \nonumber\\
	\leq{} & \p{\attack \textrm{ fails on $\hat{t}$} \land G_{\alpha} \land G_{\gamma} \land G_{S}} + \rho + e^{-k/4} + e^{-k/6}
\end{align}

Thus, to complete the proof, it suffices to show that
\begin{align*} 
\p{\attack \textrm{ fails on $\hat{t}$} \land G_{\alpha} \land G_{\gamma} \land G_{S}}
 ={} &\p{(\attack \textrm{ fails on $\hat{t}$}) \land (\textrm{$\hat{t}$ is $\alpha$-accurate}) \land (q_{(k)} \geq \gamma) \land (|S_{\lambda}| \leq d_{\lambda})} \notag \\
\leq{} &\p{(\attack \textrm{ fails on $\hat{t}$}) \land (\hat{t} \subseteq S_{\lambda}) \land (
|S_{\lambda}| \leq d_{\lambda})} \notag \\
\leq{} &\p{\left(\exists v \in \binom{S_{\lambda}}{k}~~\attack \textrm{ fails on $v$}\right) \land (|S_{\lambda}| \leq d_{\lambda})} 
\end{align*}
where we have abused notation and written $\hat{t} \subseteq S_{\lambda}$ to mean that $\hat{t}_j = 1 \Longrightarrow j \in S_{\lambda},$ and used $v \in \binom{S_{\lambda}}{k}$ to mean that $v$ is a subset of $S_{\lambda}$ of size exactly $k$.

We will now upper bound
$
\p{(\exists v \in \binom{S_{\lambda}}{k}~~\attack \textrm{ fails on $v$}) \land (|S_{\lambda}|\leq d_{\lambda})}.
$
Observe that, since the columns of $X$ are identically distributed, this probability is independent of the specific choice of $S_{\lambda}$ and depends only on $|S_{\lambda}|$.  Further, decreasing the size of $S_{\lambda}$ only decreases the probability.  Thus, we will fix a set $S$ of size exactly $d_{\lambda}$ and assume $S_{\lambda} = S$.  Thus, for our canonical choice of set $S = \set{1,\dots,d_{\lambda}}$, we need to bound
$
\p{\exists v \in \binom{S}{k}~~\attack \textrm{ fails on $v$}}.
$

\medskip

Consider a fixed vector $v \subseteq S$.  That is, a vector $v \in \zo^{d}$ such that $v_j = 1 \Longrightarrow j \in S$.  Define the event $E_{i,v}$ to be the event that $\langle x_i, v \rangle$ is too small for some specific row $i$ and some specific vector $v \subseteq S$.  That is,
$$
E_{i,v}:= \left\{ \iprod{x_i}{v} < \tau_c := k\lambda - \sqrt{2k\ln(1/\rho)}\right\}.
$$

\jon{The event should be defined in terms of $\tau_c$, right?  The conclusion of the lemma is stated in terms of $\tau_c$}
\mitali{Yeah, that was a typo.}

Since the columns of $X_{S}$ are independent, for a fixed $i$ and $v$, by Hoeffding's inequality gives
$$
\p{E_{i,v}} = \p{\textstyle\sum_{j : v_j = 1} x_{i,j} < \tau_c} \leq \rho.
$$

We have proved that the probability that $\langle x_i, v \rangle$ is small, is small for a given row.  We want to bound the probability that $\langle x_i, v \rangle$ is small for an entire set of rows $R \subseteq [n]$.  Unfortunately, since we require that $q_j \geq \lambda$ for every column $j \in S$, the rows $x_i$ are no longer independent.  However, the rows satisfy a negative-dependence condition, captured in the following claim.


\begin{claim}~\label{corr_rows}
For every $R \subseteq [n]$,
	$$\pr{}{ \bigwedge_{i \in R} E_{i,v}} \leq  \rho^{|R|}.$$
\end{claim} 
\noindent To maintain the flow of the analysis, we defer the proof of this claim to Section~\ref{proofofclaim:corr_rows}


\medskip

By definition, $\attack$ fails on $v$ only if there exists a set $R$ of exactly $c n = e^2 \rho n$ rows such that $\bigwedge_{i \in R} E_{i,v}$.  Taking a union bound over all such sets $R$ and all $v$ , we have
\begin{align*}
\p{\exists v \in \binom{S}{k}~~\attack \textrm{ fails on $v$}}
\leq{} &\binom{d_{\lambda}}{k}\cdot \binom{n}{c n} \cdot \rho^{c n} \\
\leq{} &\left( \frac{ed_{\lambda}}{k} \right)^{k} \cdot \left( \frac{en\rho}{c n} \right)^{c n} \\
\leq{} &d_{\lambda}^{k} \cdot e^{-cn}
\end{align*}
where we have used the identity $\binom{a}{b} \leq (\frac{e a }{ b})^b$.  We have already set the parameter $\lambda$, and set $d_{\lambda} = 2de^{-\frac{\lambda^2n}{2}}$.  Thus, all that remains is to show that for our choice of parameters $d_{\lambda}^k \cdot e^{-cn} \leq \rho$, which is equivalent to $cn \geq \ln(1/\rho) + k\ln(d_{\lambda})$.  Substituting our choice of $\lambda$ gives the condition
$$
\frac{k\lambda^2n}{2} \geq \ln(1/\rho) + k\ln(2d) -cn
$$
One can check that, for our choice of $n = \frac{C_3^2k\ln(2d)}{8\ln(1/\rho)}$, and our choice of $\lambda = C_3\sqrt{\frac{\ln(2d)}{n}}$ where $C_3$ has been defined above, the preceding equation is satisfied.

Thus, we have established that 
$$
\p{\exists v \in \binom{S}{k}~~\attack \textrm{ fails on $v$}}
\leq d_{\lambda}^{k} \cdot e^{-cn} \leq \rho.
$$
As we have argued above, this implies that
$$
\p{\attack \textrm{ fails on $\hat{t}$}} \leq 2\rho + e^{-k/4} + e^{-k/6}
$$
This completes the proof of the completeness lemma.

\end{proof}

\subsection{Proof of Claim~\ref{corr_rows}} \label{proofofclaim:corr_rows}

\jon{The section is already called "proof of claim 4.5" so I took out the proof environment.}

\jon{I put in a preamble so that we wouldn't need to keep writing $\land (q_j \geq \lambda, \forall j)$ everywhere.  I think it makes the notation a bit less cluttered.}

	Recall that, for a given $X \in \pmo^{n \times d}$, $E_{i,v}$ is the event that $\langle x_i, v \rangle < \tau_c$ for a specific row $i$ and a specific vector $v \subseteq S$, where $S = S_\lambda$ is the set of columns $j$ of  $X$ such that $q_j \geq \lambda$.  Thus, we can think of $X_{S} \in \pmo^{n \times |S|}$ as a matrix with $|S|$ independent columns that are uniformly random subject to the constraint that each column's mean is at least $\lambda$.  Since, flipping some entries of $X_{S}$ from $-1$ to $+1$ can only increase $\langle x_i, v \rangle$, we will in fact use the distribution $\widetilde{X}_{S}$ in which each column's mean is exactly $\lambda n$.  Thus, when we refer to the probabilities of events involving random variables $x_{i,j}$, we will use this distribution on $X_{S}$ as the probability space.  Additionally, since $v$ is fixed, and the probability is the same for all $v$, we will simply write $E_i$ to cut down on notational clutter.
	
For a specific set $R \subseteq [n]$, we need to calculate $$\pr{X_{S}}{\bigwedge_{i \in R} E_{i}}.$$
We can write
	\begin{align}\label{last_label}
	\pr{\widetilde{X}_{S}}{\underset{i \in R}{\bigwedge} E_{i}}
	={} &\pr{\widetilde{X}_{S}}{\underset{i \in R}{\bigwedge} \left(\underset{j \in v}{\sum}x_{ij} < \tau_c \right)} \nonumber  \\
	\leq{} &\pr{\widetilde{X}_{S}}{\underset{i \in R, j \in v}{\sum} x_{ij} < |R| \tau_c }.
	\end{align}
	
	The key property of $X_{S}$ is that its entries $x_{ij}$ are positively correlated.  That is, for every set $I \subset [n] \times S$ of variables $x_{ij}$, we have,
	\begin{align}\label{neg_dep}
	\pr{\widetilde{X}_{S}}{\forall (i,j) \in I~~~x_{ij} = -1} \leq \prod_{(i,j) \in I} \pr{\widetilde{X}_{S}}{x_{ij} = -1}. 
	\end{align}
	
	Since the columns of $\widetilde{X}_{S}$ are independent if we partition the elements of $I$ into sets $I_1,\ldots,I_k$, where each set $I_l$ has pairs of $I$ which come from the column $l$, then, $$\pr{\widetilde{X}_{S}}{\forall (i,j) \in I~~~x_{ij} = -1} = \prod_{l \in [k]} \pr{\widetilde{X}_{S}}{\forall (i,j) \in I_l~~~x_{ij} = -1}.$$ So it is enough to show that equation~\ref{neg_dep} holds when $I = \set{(i_1,l),\ldots,(i_p,l)}$. For simplicity of notation we will refer to these elements of $I$ as $\set{1,\ldots,p}$. We have that,
	\begin{align}\label{b1}
	\pr{}{\forall a \in I~~~x_{a} = -1} = \prod_{a = 1}^p \pr{}{x_{a} = -1 \given \left( \forall b \in \set{1,\dots,a-1}, x_b = -1 \right)}. 
	\end{align}
	We will show that each of the terms in the product is smaller than $\pr{}{x_{a} = -1}$. For a fixed $a \in I$, let $B$ be the set $\set{1,\dots,a-1}$ and let $E$ be the event that $(\forall b \in B, x_b = -1)$.  Since every column of $X_{S}$ sums to $n \lambda$, we have
	$$
	\E\left[\underset{i \in [n]}{\sum} x_{il} \given B\right] = n\lambda.
	$$ 
	On the other hand, since the bits in $B$ are all set to $-1$ and all the other bits in column $l$ are equal in expectation, 
	$$
	\E\left[\underset{i \in [n]}{\sum} x_{il} \given B\right] = -|B| + (n - |B|) \cdot \left(\E\left[x_a \given  B\right]\right),
	$$ 
	which means that
	$$\E[x_a \given B] \geq \lambda = \E[x_a ].$$  Since ${\pr{ }{x_a = -1} = (1 - \E[x_a])/2}$, we get that ${\pr{ }{x_a = -1 \given B} \leq \pr{ }{x_a = -1}}.$ Substituting this back into~\eqref{b1}, we get that the variables are positively correlated.
	
	We have that, $$\E\left[\sum_{i \in R, j \in v} x_{ij}\right] = |R|k\lambda,$$ and since Hoeffding's inequality applies equally well to positively-correlated random variables~\cite{PS97}, we also have
	$$\pr{\widetilde{X}_{S}}{ \underset{i \in R, j \in v}{\sum} x_{ij} \leq |R|\tau_c} \leq \pr{\widetilde{X}_{S}}{\sum_{i \in R, j \in V} x_{ij} < |R|k\lambda - |R|\sqrt{2k\ln(1/\rho)}} \leq \exp\left(-\frac{\left(|R|\sqrt{2 k \ln(1/\rho)}\right)^2}{2 |R| k} \right) = \rho^{|R|}.$$
	
	Substituting this in equation~\ref{last_label}, we get that,
	$$	\pr{\widetilde{X}_{S}}{\underset{i \in R}{\bigwedge} E_{i} } \leq \rho^{|R|}.$$  Finally, we use the fact that, by our definition of the distributions $X_{S}, \widetilde{X}_{S}$, we have
	$$
	\pr{X_{S}}{\bigwedge_{i \in R} E_{i}} \leq \pr{\widetilde{X}_{S}}{\bigwedge_{i \in R} E_{i}} \leq \rho^{|R|}.
	$$
	This completes the proof.